\documentclass[english,aps,pra,superscriptaddress,floatfix,notitlepage,reprint,pdftex,unicode=true,colorlinks=true,citecolor=Blue,linkcolor=RubineRed,urlcolor=Blue]{revtex4-2}
\usepackage[T1]{fontenc}
\usepackage[latin9]{inputenc}
\usepackage[toc,page,title,titletoc,header]{appendix}
\usepackage{color}

\usepackage{float}
\usepackage{babel}
\usepackage{comment}
\usepackage{mathrsfs}
\usepackage{bm}
\usepackage{amsmath}
\usepackage{amsthm}
\usepackage{amssymb}
\usepackage{algorithm}
\usepackage{algpseudocode}
\newenvironment*{GeneralCase}
    {\par\noindent\textbf{General Case:}\itshape} % 开始环境时的设置
    {\par} % 结束环境时的设置
\newenvironment*{QubitCase}
    {\par\noindent\textbf{Qubit Case:}\itshape} % 开始环境时的设置
    {\par} % 结束环境时的设置

\newtheorem{observation}{Observation}
\newtheorem*{proposition}{Proposition}

\usepackage{cancel}
\usepackage{graphicx}
\usepackage[unicode=true]
 {hyperref}

\makeatletter
\theoremstyle{plain}
\newtheorem{thm}{\protect\theoremname}
\theoremstyle{definition}
\newtheorem{example}{Example}

\usepackage{braket}
\usepackage{txfonts} 
\usepackage{graphicx}
\usepackage[usenames,dvipsnames]{xcolor}
\date{\today}

\usepackage{yhmath}

\theoremstyle{slantedtheorem}
\newtheorem*{lemma*}{Lemma}

\newtheorem{corollary}{Corollary}

\makeatother

\providecommand{\theoremname}{Theorem}

\begin{document}
\title{Optimal Local  Measurements in Single-Parameter Quantum Metrology}
\author{Jia-Xuan Liu}
\thanks{These two authors contributed equally.}
\affiliation{Hefei National Research Center for Physical Sciences at the Microscale
and School of Physical Sciences, Department of Modern Physics, University
of Science and Technology of China, Hefei, Anhui 230026, China}
\author{Jing Yang\href{https://orcid.org/0000-0002-3588-0832}{\includegraphics[scale=0.05]{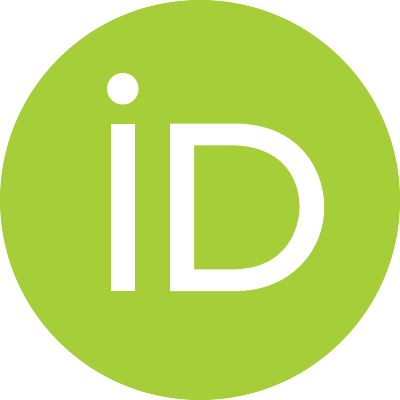}}}
\thanks{These two authors contributed equally}
\email{jing.yang@su.se}

\affiliation{Nordita, KTH Royal Institute of Technology and Stockholm University,
Hannes Alfv\'ens vag 12, 106 91 Stockholm, Sweden}
\author{Hai-Long Shi}
\affiliation{QSTAR, INO-CNR, and LENS, Largo Enrico Fermi 2, 50125 Firenze, Italy}
\affiliation{Innovation Academy for Precision Measurement Science and Technology,
Chinese Academy of Sciences, Wuhan 430071, China}
\affiliation{Hefei National Laboratory, University of Science and Technology of
China, Hefei 230088, China}
\author{Sixia Yu}
\email{yusixia@ustc.edu.cn}

\affiliation{Hefei National Research Center for Physical Sciences at the Microscale
and School of Physical Sciences, Department of Modern Physics, University
of Science and Technology of China, Hefei, Anhui 230026, China}
\affiliation{Hefei National Laboratory, University of Science and Technology of
China, Hefei 230088, China}
\begin{abstract}
Quantum measurement plays a crucial role in quantum metrology. 
Due to the limitations of experimental capabilities, collectively measuring multiple copies of probing systems can present significant challenges. Therefore, the concept of locality in quantum measurements must be considered. In this work, we investigate the possibility of achieving the Quantum Cram\'er-Rao Bound (QCRB) through local measurements (LM). 
We first demonstrate that
if there exists a LM to saturate the QCRB for qubit systems, then we can construct another rank-1 local projective measurement to saturate the QCRB.
In this sense, rank-1 local projective measurements are sufficient to analyze the problem of saturating the QCRB.
For pure qubits, we propose two necessary and sufficient methods to determine whether and how a given parameter estimation model can achieve QCRB through LM. 
The first method, dubbed iterative matrix partition method (IMP) and based on unitary transformations that render the diagonal entries of a tracless matrix vanish, elucidates the underlying mathematical structure of LM as well as the local measurements with classical communications (LMCC), generalizing the result by \href{http://10.1088/2058-9565/ab71f8}{[Zhou et al Quantum Sci. Technol. 5, 025005 (2020)]}, which only holds for the later case.
We clarify that the saturation of QCRB through LM for the GHZ-encoded states is actually due to the self-similar structure in this approach.
The second method, dubbed hierarchy of orthogonality conditions (HOC) and based on the parametrization of rank-1 measurements for qubit systems,  allows us to construct several examples of saturating QCRB, including the three-qubit W states and $N$-qubit W states ($N \geq 3$).
Our findings offer insights into achieving optimal performance in quantum metrology when measurement resources are limited.
\end{abstract}
\maketitle

\renewcommand{\figurename}{FIG.}

\section{introduction}
Locality plays a crucial role in various
branches of physics, encompassing high energy physics \cite{Sterman1993,huang2010quantum}, condensed matter physics \cite{PhysRevB.72.045141,PhysRevB.82.155138} and quantum information
theory \cite{doi:10.1126/science.1121541,yang2022minimumtimequantumcontrolquantum,(Anthony)Chen_2023}. In the context of many-body systems, locality gives rise to the Lieb-Robinson bound \cite{Lieb1972,Nachtergaele_2008,PhysRevLett.97.050401}, which sets
an upper limit on the spread of local operators. Despite the recent resurgence of interest in quantum metrology using many-body Hamiltonians \cite{PhysRevLett.98.090401,PhysRevLett.100.220501,PhysRevLett.119.010403,PhysRevResearch.4.013133,PhysRevLett.128.160505}, the investigation of locality in
the sensing Hamiltonian has only been undertaken until recently \cite{PhysRevLett.132.100803,Yin2024heisenberglimited,PhysRevLett.130.170801}.

On the other hand, it is well known in quantum metrology theory that the Quantum Cram\'er-Rao Bound (QCRB) sets an upper limit on the achievable precision of measurements. For single-parameter estimation, the eigenstates of the symmetric logarithmic derivative (SLD) operator serve as the optimal measurements, and this limit is always achievable. However, the high degeneracy of the SLD operator, along with the fact that in most cases the eigenspace consists of non-local states, makes high-precision experimental implementation a significant challenge. A more practically relevant question is whether there exists a specific SLD operator such that all its eigenstates are local (i.e., completely separable), thereby forming an local optimal measurement. We refer to a metrological model that meets this requirement as \textit{locally optimal achievable} or having \textit{local optimal achievability}. Considering a non-interacting encoding Hamiltonian $H_{\lambda} = \lambda \sum_{j} Z_{j}$, where $Z_{j}$ is the Pauli $Z$ operator acting on the $j$-th qubit. Ref. \cite{PhysRevLett.96.010401} points out that when the system is prepared in a Greenberger-Horne-Zeilinger (GHZ) state, it attains maximum precision and is also locally optimal achievable. However, to our knowledge, little is known about whether other metrological models can also be locally optimal achievable. Furthermore, for pure states, Zhou et al~\cite{Zhou_2020} demonstrated that rank-1 and projective local  measurements with classical communication (LMCC) can be constructed to saturate the QCRB. However, taking into account classical communication among particles, the total number of optimal measurement bases for general high-dimensional encoding states, such as $N$-qubit states, grows exponentially with the number of particles. This could require exponentially large experimental resources, making it often infeasible. In summary, local measurements in quantum metrology remain largely uncharted at both fundamental and experimental levels.
  
In this paper, we provide a comprehensive study of identifying QCRB-saturating local  measurements (LM) for a given parameter estimation model. In contrast with LMCC, LM does not require the classical communications between different local observables and can reduce the measurement resources dramatically, if it exists .  We first determine whether the model is locally optimal achievable; if it is, we then provide the corresponding local optimal  measurements. By applying our theory, we demonstrate the local optimal achievability in various metrological scenarios, including the case of two-qubit encoded pure states, where the characteristic matrix $M$ that defines the saturation of the QCRB based on encoded states is either a diagonal matrix (with all off-diagonal elements being zero) or a zero-diagonal matrix (with all diagonal elements being zero), as well as some special encoded cases using W states as probe states. We construct an example where local optimal achievability is not possible, reconstructing the existing conclusion that LMCC can universally reach optimality and pointing out that the universality of local optimal achievability can be maintained at most for two qubits in qubit systems.

The organization of this paper is as follows: We first review the background and basic concepts of quantum metrology and quantum measurements in Sec. \ref{section 2}. In Sec. \ref{section 3A}, we prove that for qubit systems, If a metrological model is locally optimal achievable, we can always construct a rank-1 local optimal measurement. In Sect. \ref{section 3b}, we describe how the first  method, dubbed iterative matrix partition (IMP) method, can be used to study local optimal achievability. In Sect. \ref{section 3c}, we point out that the IMP method, combined with a simple observation, can easily reconstruct the optimal achievability of LMCC given in Ref. \cite{Zhou_2020}. In Sect. \ref{section 3d}, we demonstrate how the IMP method reveals the limitations of universally saturating the QCRB in the LM scenario, providing two examples related to the structure of the matrix $M$.  In Sec.~\ref{section 4a} and Sec.~\ref{section 4b}, we introduce a second method, dubbed hierarchy of orthogonality conditions (HOC) method, for investigating local optimal achievability within the general encoding framework and the unitary phase encoding framework, respectively. In Sec. \ref{section 4c}, we present several applications of the HOC method, including the three-qubit W state and the $N$-qubit W state. Finally, we summarize our work in Sec. \ref{conclusion}.

\section{BACKGROUND}\label{section 2}

In quantum metrology, the general procedure for parameter estimation begins with the preparation of a probe state $\rho_{0}$, which is independent of the parameter to be estimated. After undergoing parameter-dependent dynamics $\mathcal{L}_{\lambda}$, we obtain the encoded state $\rho(\lambda) = \mathcal{L}_{\lambda}(\rho_{0})$. Following this, a positive operator-valued measurement (POVM) is performed using a set of positive semi-definite Hermitian operators $\left\{E_{x} \equiv K_{x}^{\dagger}K_{x}\right\}$ that satisfy $\sum_{x} E_{x} = I$ \cite{Liu_2020}. Based on the measurement results, we obtain an estimator for the parameter $\lambda$, denoted as $\hat{\lambda}_{\mathrm{est}}$. If we label the mean value of the estimator as $\langle \hat{\lambda}_{\mathrm{est}} \rangle$ and its variance as $\left(\Delta \lambda_{\mathrm{est}}\right)^{2}$, then for an unbiased estimation, we have $\langle \hat{\lambda}_{\mathrm{est}} \rangle - \lambda = 0$,  
\begin{equation}\label{QCR bound}
    \left(\Delta \lambda_{\mathrm{est}}\right)^{2}\ge \frac{1}{m\mathcal{F}_{C}}\ge \frac{1}{m\mathcal{F}_{Q}},
\end{equation}
where 
\begin{equation}
    F_{C}(\lambda)  \equiv \sum_{x: p(x |\lambda) \neq 0}\frac{\left[\partial_{\lambda} p(x | \lambda)\right]^{2}}{p(x | \lambda)}+\sum_{x:p(x|\lambda)=0}\lim_{p(x|\lambda) \to 0} \frac{\left[\partial_{\lambda} p(x | \lambda)\right]^{2}}{p(x | \lambda)}
\end{equation}
is the classical Fisher information (CFI), which quantifies the parameter information in the conditional probability $p(x|\lambda)=\operatorname{Tr}\left[\rho(\lambda)E_{x}\right]$ resulting from a POVM on $\rho(\lambda)$. $m$ denotes the number of repetitions of the experiment, and
\begin{equation}
\mathcal{F}_{Q}\equiv \sum_{x}\operatorname{Tr}\left[E_{x}L_{\lambda}\rho(\lambda) L_{\lambda}\right]=\operatorname{Tr}\left[\rho(\lambda) L^{2}_{\lambda}\right]
\end{equation}
is the quantum Fisher information (QFI), where we define the SLD operator $L_{\lambda}$, which satisfies $L_{\lambda}^{\dagger} = L_{\lambda}$, as
\begin{equation}\label{SLD op}
      \frac{\partial \rho(\lambda)}{\partial \lambda} = \frac{1}{2} \left[\rho(\lambda)L_{\lambda} + L_{\lambda}\rho(\lambda)\right].
  \end{equation}

The first bound in Eq. \eqref{QCR bound} is saturated in the asymptotic limit as $m \to \infty$, while the second bound, known as the Quantum Cram\'er-Rao Bound (QCRB) \cite{PhysRevLett.72.3439}, depends on the chosen quantum measurement for its saturation. Measurements that achieve the saturation of the QCRB are referred to as optimal measurements. Previous results indicate that the projective measurements formed by the eigenstates of $L_{\lambda}$
  saturate the QCRB. However, for any Hermitian operator $S$,  considering that Eq. \eqref{SLD op} remains invariant under the transformation $L_{\lambda} \to L_{\lambda} + (1 - P_{\rho_{\lambda}})S(1 - P_{\rho_{\lambda}})$, where $P_{\rho_{\lambda}}$ is the projector to the support of $\rho_{\lambda}$. This suggests that $L_{\lambda}$
  is not unique, and thus the optimal measurements that saturate the QCRB are also not unique. Given the experimental challenges associated with non-local (collective) measurements, our attention in what follows will be focused on the existence of LM or LMCC.

Throughout this work, we shall consider qubit systems. 
For $N$-qubits, $\left\{E_{x}\right\}$ is called a LM if we can
decompose each $E_{x}$ as
\begin{equation}
E_{x}=\bigotimes_{i=1}^{N}\mathcal{E}_{x_{i}}^{(i)},
\end{equation}
where $\mathcal{E}_{x_{i}}^{(i)}$ is the measurement operator that acts locally on the $i$-th qubit and outputs $x_{i}$, satisfying $\sum_{x_{i}} \mathcal{E}_{x_{i}}^{(i)} = I^{(i)}$. LMCC refers to allowing classical communication between different qubits, where the measurement choice of the $i$-th qubit can depend on the measurement outcomes of the first $i-1$ qubits, i.e.,
\begin{equation}
E_{x}=\mathcal{E}_{x_{1}}^{(1)}\bigotimes_{i=2}^{N} \mathcal{E}_{x_{i};x_{1}\cdots x_{i-1}}^{(i)}.
\end{equation}

The necessary and sufficient condition for QCRB saturation, which was first provided by Caves \cite{PhysRevLett.72.3439}, was later reformulated by Zhou et al. \cite{Zhou_2020} into a new, more easily computable equivalent form using the Choi-Jamiolkowski duality \cite{PhysRevA.87.022310}. However, since when $p(x|\lambda) = 0$, applying L'H\^opital's rule shows that $\lim_{p(x|\lambda) \to 0} \frac{\left[\partial_{\lambda} p(x | \lambda)\right]^{2}}{p(x | \lambda)} = \operatorname{Tr} E_{x} L_{\lambda} \rho L_{\lambda}$ automatically saturates the QCRB without requiring additional conditions \cite{PhysRevA.100.032104}. Thus, for a pure state $\rho(\lambda) = |\psi(\lambda)\rangle \langle \psi(\lambda)|$, the condition that saturates the QCRB can be simplified to, compared to Zhou et al~\cite{Zhou_2020}: \textit{the $\mathrm{QCRB}$ is saturated if and only if 
\begin{equation}\label{ExMEx}
    E_{x} M E_{x} = 0
\end{equation}
holds for $\forall x$, where $M \equiv \left[|\psi(\lambda)\rangle\langle\psi(\lambda)|, L_{\lambda}\right]$.}

\begin{figure*}[t]
\begin{centering}
\includegraphics[scale=0.25]{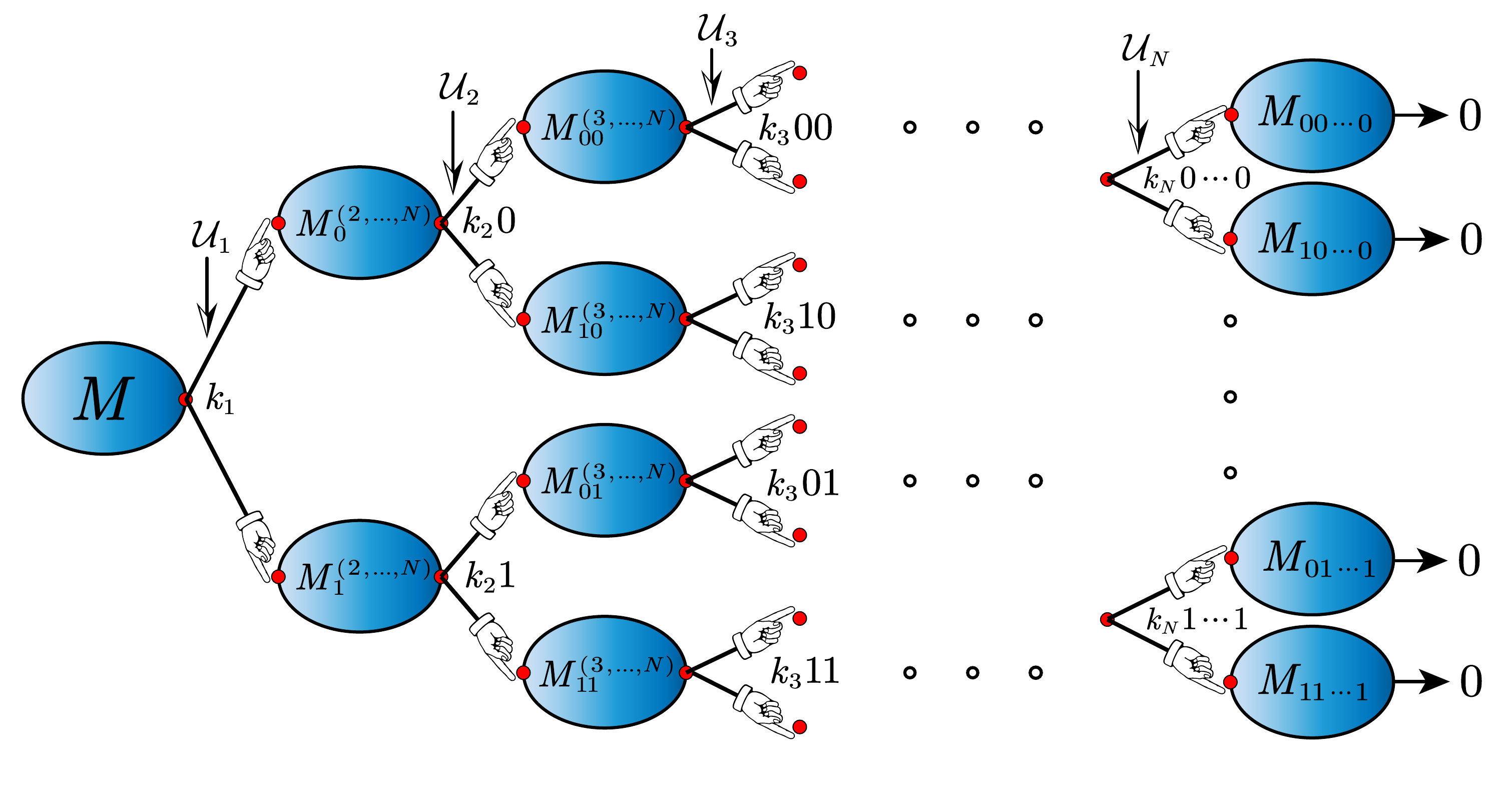}
\par\end{centering}
\caption{\label{fig:IMP}Schematic illustration of the iterative matrix partition (IMP) method, see detailed description in Sec.~\ref{section 3}. }
\end{figure*}

\vspace{0.9cm}
\section{IMP METHED}\label{section 3}
While we do not impose any additional restrictions on the selected POVM a priori, in practice, considering only rank-$1$ POVM is sufficient. This is because for each rank-$n$ ($n\ge 2$) POVM operator $E_{x}=\sum_{i=1}^{n}e_{x,i}E_{x,i}=\sum_{i=1}^{n}e_{x,i}|\Pi_{x,i}\rangle \langle \Pi_{x,i}|$, we can construct $n$ new rank-$1$ POVM operators $E_{x,i}^{\prime}\equiv e_{x,i}E_{x,i}$ as replacements. Clearly, $E_{x}ME_{x}=0$ implies that $E_{x,k}^{\prime}ME_{x,k}^{\prime}=0$, $\forall  k$, which means that rank-$1$ POVM are enough to consider~\cite{PhysRevA.100.032104}. Furthermore, if there exists a local POVM that saturates the QCRB and includes rank-$n$ ($n\ge 2$) POVM operators, we can always decompose the rank-$n$ POVM operators into $n$ rank-$1$ local POVM operators, thereby providing a new rank-$1$ local POVM that saturates the QCRB. In other words, in the context of LM, considering rank-$1$ POVM is sufficient. Now we shall further argue for qubit systems, it is sufficient to consider rank-$1$ projective measurements.

\subsection{Rank-$1$ projective LM are sufficient for qubit systems}\label{section 3A}

For an $N$-qubit system, we observe that if there exists an local optimal non-projective measurement (outcome greater than $2$), then there must simultaneously exist binary local optimal  projective measurements, which means that considering rank-$1$ local projective measurements is sufficient. To see this, first note that a local rank-$1$ POVM can be written as $E_{x}=\bigotimes _{i=1}^{N}\mathcal{E}_{x_{i}}^{(i)}$, $\mathcal{E}_{x_{i}}^{(i)}=a_{x_{i}}^{(i)}\left(I^{(i)}+\boldsymbol{n}_{x_{i}}^{(i)} \cdot \boldsymbol{\sigma}^{(i)}\right) / 2 $ where $\boldsymbol{n}_{x_{i}}^{(i)}$ is the Bloch vector for the $i$-th qubit, $\boldsymbol{\sigma}^{(i)} \equiv (X_{i}, Y_{i}, Z_{i})$, and $I^{(i)}$ represent the Pauli operators and the identity operator for the $i$-th qubit, respectively. Then we obtain 
\begin{equation}
    \begin{array}{l}
a_{x_{i}}^{(i)} \neq 0,|\boldsymbol{n}_{x_{i}}^{(i)}|=1, \:\forall i.\\
\sum_{x_{i}} \mathcal{E}_{x_{i}}^{(i)}=I^{(i)} \Rightarrow \sum_{x_{i}} a_{x_{i}}^{(i)}=2, \sum_{x_{i}} \boldsymbol{n}_{x_{i}}^{(i)}=0 .
\end{array}
\end{equation}
The condition for QCRB-saturating local measurement reads
\begin{equation}
    \left[\bigotimes_{i=1}^{N} \mathcal{E}_{x_{i}}^{(i)}\right] M\left[\bigotimes_{i=1}^{N} \mathcal{E}_{x_{i}}^{(i)}\right]=0\Rightarrow \operatorname{Tr} M\left[\bigotimes_{i=1}^{N} \mathcal{E}_{x_{i}}^{(i)}\right]=0.
\end{equation}
For the $i$-th qubit, by summing over all the other $x_{j}$ ($j\neq i$), we have
\begin{equation}\label{TrMrsigma}
    \frac{1}{2} \operatorname{Tr} M a_{x_{i}}^{(i)}\left(I^{(i)}+\boldsymbol{n}_{x_{i}}^{(i)} \cdot \boldsymbol{\sigma}^{(i)}\right) =0\Rightarrow \operatorname{Tr} M\boldsymbol{n}_{x_{i}}^{(i)} \cdot \boldsymbol{\sigma}^{(i)}=0,
\end{equation}
where we have used the result $\operatorname{Tr}M=0$. Further considering qubits $k$ and $l$, combining with Eq. \eqref{TrMrsigma}
 we can show that
\begin{equation}
    \operatorname{Tr} M \mathcal{E}_{x_{k}}^{(k)}\otimes  \mathcal{E}_{x_{l}}^{(l)}=0 \Rightarrow \operatorname{Tr} M \left(\boldsymbol{n}_{x_{k}}^{(k)} \cdot \boldsymbol{\sigma}^{(k)}\right)\left( \boldsymbol{n}_{x_{k}}^{(l)} \cdot \boldsymbol{\sigma}^{(l)}\right)=0.
\end{equation}
By the induction method, 
\begin{equation}\label{TrMrsigma0}
    \operatorname{Tr} M \bigotimes_{k \in \alpha} \boldsymbol{n}_{x_{k}}^{(k)} \cdot \boldsymbol{\sigma}^{(k)}=0, \quad \forall \alpha \subseteq\{1,2, \ldots, N\}, \quad \forall x_{k}.
\end{equation}
Define local projector as $\tilde{\mathcal{E}}_{ \pm}^{(i)}=\left(I^{(i)} \pm \boldsymbol{n}_{1}^{(i)} \cdot \boldsymbol{\sigma}^{(i)}\right) / 2$, then by condition \eqref{TrMrsigma0} we can show that
\begin{equation}
    \operatorname{Tr} M\bigotimes_{i=1}^{N} \tilde{\mathcal{E}}_{y_{i}}^{(i)}=0 \Rightarrow\left[\bigotimes_{i=1}^{N} \tilde{\mathcal{E}}_{y_{i}}^{(i)}\right] M\left[\bigotimes_{i=1}^{N} \tilde{\mathcal{E}}_{y_{i}}^{(i)}\right]=0, \: \forall y_{i},
\end{equation}
which confirms that $\tilde{\mathcal{E}}_{y}\equiv \bigotimes _{i=1}^{N}\tilde{\mathcal{E}}_{y_{i}}^{(i)}$ can be replaced as the new QCRB-saturating LM.

\subsection{IMP for LM}\label{section 3b}
If there exists a rank-$1$ projective LM $\mathcal{E}_{x_{i}}^{(i)}=|\pi_{x_{i}}^{(i)}\rangle\langle \pi_{x_{i}}^{(i)}|$, where $x_{i}\in \{0,1\}$, that saturates the QCRB, then it satisfies $\left[\bigotimes_{i=1}^{N} \mathcal{E}_{x_{i}}^{(i)}\right] M\left[\bigotimes_{i=1}^{N} \mathcal{E}_{x_{i}}^{(i)}\right]=0$ for $\forall x_{i}$. Defining the unitary matrix $\mathcal{U}_{i} \equiv (|\pi_{0}^{(i)}\rangle,|\pi_{1}^{(i)}\rangle )$, we have 
\begin{equation}\label{UMU0}
    \left(\left[\bigotimes_{i=1}^{N} \mathcal{U}_{i}\right]^{\dagger} M\left[\bigotimes_{i=1}^{N} \mathcal{U}_{i}\right]\right)_{k k}=0,   \forall k,
\end{equation}
where $O_{kk}$ represents the $k$-th diagonal element of the matrix $O$. On the other hand, if Eq. \eqref{UMU0} holds, then defining local projectors from $\mathcal{U}_{i}=\left(|\eta_{0}^{(i)}\rangle,|\eta_{1}^{(i)}\rangle\right)$ and setting $|\pi_{x_{i}}^{(i)}\rangle=|\eta_{x_{i}}^{(i)}\rangle$ immediately restores condition $\left[\bigotimes_{i=1}^{N} \mathcal{E}_{x_{i}}^{(i)}\right] M\left[\bigotimes_{i=1}^{N} \mathcal{E}_{x_{i}}^{(i)}\right]=0$. Starting from Eq. \eqref{UMU0}, based on $\operatorname{Tr}M=0$ and the matrix trace remains invariant under local unitary transformations, we have developed an operational necessary and sufficient condition, which we refer to as ``\textit{iterative matrix partition}'' (IMP), to investigate the  QCRB-saturating LM in qubit systems: \textit{After the local unitary transformation, the diagonal elements of the resulting $\mathit{2\times 2}$
block matrix are still traceless}.

Specifically, the matrix $M$ undergoes a unitary transformation with respect to the $1$-st qubit, resulting in 
\begin{equation}
\mathcal{U}_{1}^{\dagger} M \mathcal{U}_{1}=\left(\begin{array}{cc}
M_{ 0}^{(2, \ldots, N)} & \times \\
\times & M_{1}^{(2, \ldots, N)}
\end{array}\right)_{[1]},
\end{equation}
where $M_{0}^{(2, \ldots, N)}$ and $M_{1}^{(2, \ldots, N)}$ are matrices defined on the subspace of qubits $\{2, 3,\ldots, N\}$, and the subscript $[j]$ denotes the block matrix representation on $j$-th qubit. Let the  diagonal block matrix $M_{k_{1}}^{(2, \ldots, N)}$, where $k_{1}\in\{0,1\}$, undergo a unitary transformation with respect to the $2$-nd qubit, yielding
\begin{equation}
    \mathcal{U}_{2}^{\dagger} M_{k_{1}}^{(2, \ldots, N)} \mathcal{U}_{2}=\left(\begin{array}{cc}
M_{k_{1} 0}^{(3, \ldots, N)} & \times \\
\times & M_{k_{1}1}^{(3, \ldots, N)}
\end{array}\right)_{[2]},
\end{equation}
where $M_{k_{1}0}^{(3, \ldots, N)}$ and $M_{k_{1}1}^{(3, \ldots, N)}$ are matrices defined on the subspace of qubits $\{3,4, \ldots, N\}$. Following this procedure iteratively until $N$ local unitary transformations are applied, we obtain $M_{k_{1}k_{2}\ldots k_{N}}$. Based on the above, the IMP requires that all the matrices obtained, $M_{0}^{(2, \ldots, N)}
 , M_{1}^{(2, \ldots, N)}
 , M_{00}^{(3, \ldots, N)}, \ldots, M_{\underbrace{11 \cdots 1}_{N}}$, are traceless. A more intuitive illustration of the above process is shown in Fig. \ref{fig:IMP}.
 
\subsection{IMP for LMCC}\label{section 3c}
Through straightforward analysis, we provide the following observation.
%We first prove the following lemma \cite{horn2012matrix,PhysRevLett.85.4972,fillmore1969similarity}:
\begin{observation}\label{observation}
A set of traceless $\mathit{2\times 2}$ Hermitian (anti-Hermitian) matrices $\mathcal{A}_{j}=\boldsymbol{a}_{j} \cdot \boldsymbol{\sigma}$ ($\mathcal{A}_{j}=i \boldsymbol{a}_{j} \cdot \boldsymbol{\sigma}$) can be simultaneously transformed into zero-diagonal matrices by the same unitary transformation if and only if the vectors $\{\boldsymbol{a}_{j}\}$ are coplanar.
\end{observation}
The proof of Observation \ref{observation} can be found in Appendix \ref{Universality of Two-Qubit Pure States}. Based on the above observation, we immediately obtain the following corollary.
\begin{corollary}\label{corollary 1}
Any $2 \times 2$ traceless Hermitian (or anti-Hermitian) matrix is unitarily similar to a zero-diagonal matrix.
\end{corollary}
Now, based on Corollary \ref{corollary 1}, we will reproduce the results from \cite{Zhou_2020} starting from IMP: \textit{The rank-$1$ projective $\mathrm{LMCC}$ saturates the $\mathrm{QCRB}$ for pure states}. 
According to the Corollary \ref{corollary 1}, we can always find $\mathcal{U}_{1}$ such that $\operatorname{Tr} M_{0}^{(2, \ldots, N)}=\operatorname{Tr} M_{1}^{(2, \ldots, N)}=0$. For the different measurement results of the $1$-st qubit, $0$ and $1$, Corollary \ref{corollary 1} indicates that we can also find the second unitary $\mathcal{U}_{2;0}$ such that $\operatorname{Tr} M_{00}^{(3, \ldots, N)}=\operatorname{Tr} M_{01}^{(3, \ldots, N)}=0$ and $\mathcal{U}_{2;1}$ such that $\operatorname{Tr} M_{10}^{(3, \ldots, N)}=\operatorname{Tr} M_{11}^{(3, \ldots, N)}=0$, respectively. Continuing in this manner, with each iteration relying on the previous $k-1$ qubits' measurement outcomes, we simply need to unitarily transform a traceless anti-Hermitian $2\times 2$ block matrix into a zero-diagonal $2\times 2$ block matrix, which, according to the lemma, can always be achieved. Finally, we can find $1+2+2^{2}+\cdots +2^{N-1}=2^{N}-1$ local unitary operators, $\mathcal{U}_{1}, \mathcal{U}_{2;0},\mathcal{U}_{2;1},\mathcal{U}_{3;00},\ldots, \mathcal{U}_{N;\underbrace{11 \cdots 1}_{N-1}}$. They constitute rank-$1$ projective LMCC.

\subsection{The application of IMP}\label{section 3d}
Given that any two vectors in $\mathbb{R}^3$
  space originating from the origin are obviously coplanar, combined with Observation \ref{observation}, we obtain
\begin{corollary}\label{corollary 2}
    Any two Hermitian (or anti-Hermitian) matrices can be transformed simultaneously into zero-diagonal matrices through the same unitary transformation.
\end{corollary}

For pure states, the essential difference in saturating the QCRB between LM and LMCC arises in the IMP method due to multiple traceless matrices becoming zero-diagonal under the same unitary transformation, as opposed to a single traceless matrix becoming zero-diagonal under a unitary transformation. The former is only satisfied in a few cases, while the latter is more universal. For the former, we can further infer that \textit{if a single unitary matrix exists such that $\textit{K}$ $\mathit{2\times 2}$ traceless matrices can be simultaneously transformed into zero-diagonal matrices, then the $\mathrm{QCRB}$ is saturated for a pure $\mathit{\left(\log_{2}{K}+1\right)}$-qubit state}. The universal saturation property for $K=1$ has been proven by Corollary \ref{corollary 1}, corresponding to trivial single-qubit pure states and nontrivial LMCC. Corollary \ref{corollary 2} confirms the universality in the case of $K=2$, which means that two-qubit pure encoding states are always locally optimally achievable. Considering that three or more vectors in $\mathbb{R}^3$
  are not necessarily coplanar, the universality fails for $K \geq 4$. Thus, whether the encoding states of an $N$-qubit pure state system ($N \geq 3$) are locally optimally achievable depends on the encoding states themselves.

In addition, IMP can also directly identify the local optimal measurement based on the matrix structure. For instance, if $M$ is already in a zero-diagonal form, selecting $\mathcal{U}_{1} = \mathcal{U}_{2} = \cdots = \mathcal{U}_{N} = I$ would satisfy the conditions for IMP, indicating that the computational basis itself constitutes the local optimal measurement for saturating the QCRB.

Another typical class of $M$ matrices is the diagonal matrices. It is not difficult to imagine that for all $i=1,2,\ldots,N$, the block matrices obtained after the $i$-th unitary transformation, $M_{k_{1}\ldots k_{i-1}0}^{(i+1,\ldots,N)}$ and $M_{k_{1}\ldots k_{i-1}1}^{(i+1,\ldots,N)}$, remain diagonal. If a traceless matrix is also diagonal, a uniform zero-diagonal unitary transformation is the Hadamard transformation $\mathcal{H}_{i}=(X_{i}+Z_{i})/\sqrt{2}$. This is because the block diagonal matrix $S = Z_{1} \otimes A_{2\cdots N} + I^{(1)} \otimes B_{2\cdots N}$ ($A_{2\cdots N}$  and $B_{2\cdots N}$ represent operators in the space of the qubit set $\{2, 3, \ldots, N\}$), which satisfies $\operatorname{Tr} S = \operatorname{Tr} B_{2\cdots N} = 0$, transforms under local operations to $\tilde{S} = \mathcal{H}_{1} S \mathcal{H}_{1} = X_{1} \otimes A_{2\cdots N} + I^{(1)} \otimes B_{2\cdots N}$, satisfying $\operatorname{Tr} \tilde{S}_{0}^{(2,\ldots,N)} = \operatorname{Tr} \tilde{S}_{1}^{(2,\ldots,N)} =\operatorname{Tr}B_{2\cdots N}= 0$, where $\tilde{S}_{0}^{(2,\ldots,N)}$ and $\tilde{S}_{1}^{(2,\ldots,N)}$ are two block diagonal matrices represented in the basis of the $1$-st qubit. Therefore, by choosing $\mathcal{U}_{1} = \mathcal{U}_{2} = \cdots = \mathcal{U}_{N} = \mathcal{H}$, the IMP can be completed, indicating that when the $M$ matrix is diagonal, the optimal LM for saturating the QCRB is $E_{x} = \bigotimes_{i=1}^{N} |\pm^{(i)} \rangle \langle \pm^{(i)}|$, where $|\pm^{(i)}\rangle =\frac{1}{\sqrt{2}}(|0^{(i)}\rangle+|1^{(i)}\rangle)$. As we already know, the optimal LM that saturates the QCRB for the GHZ-like state $|G(\lambda)\rangle=(|00\cdots 0\rangle+e^{-iN\lambda}|11\cdots1\rangle)/{\sqrt{2}}$ is also $E_{x}=\bigotimes_{i=1}^{N}|\pm^{(i)} \rangle \langle\pm ^{(i)}|$. This consistency arises from the following theorem:
\begin{thm}\label{GHZ example}
The $M$ matrix given by a $N$-qubit entangled encoded state $|\psi(\lambda)\rangle$ is diagonal iff $|\psi(\lambda)\rangle$ is a $\mathrm{GHZ}$-type encoded state $|G(\lambda)\rangle=\frac{1}{\sqrt{2}}(|a_{1}a_{2}\cdots a_{n}\rangle+e^{-i\theta(\lambda)}|b_{1}b_{2}\cdots b_{n}\rangle)$, where $\langle a_{i}|b_{i}\rangle =0$, $\forall i$.    
\end{thm}

{\begin{proof}
Firstly, to prove the sufficiency, considering that for pure state $\rho(\lambda)=|\psi(\lambda)\rangle \langle \psi(\lambda)|$, $L_{\lambda}=2\partial_{\lambda}\rho(\lambda)$, $M$ is given by
\begin{equation}\label{pure state M matrix}
M=2\left(|\psi\rangle \langle \partial_{\lambda}\psi|-|\partial_{\lambda}\psi\rangle \langle \psi|+2\langle \psi|\partial_{\lambda}\psi\rangle|\psi\rangle\langle \psi|\right).
\end{equation}
Substituting $|\psi(\lambda)\rangle =|G(\lambda)\rangle$ into Eq. \eqref{pure state M matrix}, we obtain
\begin{equation}
    M=-i\partial_{\lambda}\theta(|a_{1}\cdots a_{n}\rangle\langle a_{1}\cdots a_{n}|-|b_{1}\cdots b_{n}\rangle\langle b_{1}\cdots b_{n}|),
\end{equation}
which is evidently diagonal. Continuing to prove necessity. Note that Eq. \eqref{pure state M matrix} can be further simplified to 
\begin{equation}\label{2cpsiperp}
M=2c\left(|\psi\rangle\langle \psi^{\perp}|-|\psi^{\perp}\rangle\langle \psi|\right),
\end{equation}
where $c=\sqrt{\langle \partial_{\lambda}\psi|\partial_{\lambda}\psi\rangle-\left|\langle \psi|\partial_{\lambda}\psi\rangle\right|^{2}}$, $|\psi^{\perp}\rangle\equiv \frac{1}{c}(1-|\psi\rangle \langle\psi|)|\partial_{\lambda}\psi\rangle$ satisfying $\langle \psi^{\perp}|\psi\rangle=0$. Since $|\psi\rangle$ and $|\psi^{\perp}\rangle$ are both rank-$1$, the eigenvalues of $M$ are $\lambda_{M}=(2ic,-2ic,0,0,\ldots,0)$. Combing this with the fact that $M$ is diagonal in the computational basis $V\equiv \left\{|00\cdots 00\rangle,|00\cdots01\rangle,\ldots,|11\cdots 11\rangle\right\}$, we conclude that 
\begin{equation}\label{2icpsi1psi2}
    M=2ic(|\psi_{1}\rangle\langle \psi_{1}|-|\psi_{2}\rangle\langle \psi_{2}|),
\end{equation}
where $|\psi_{1}\rangle\in V$, $|\psi_{2}\rangle\in V-\{|\psi_{1}\rangle\}$. By combining Eq. \eqref{2cpsiperp} and Eq. \eqref{2icpsi1psi2}, and utilizing the orthonormality relation, we obtain $\left|\left\langle \psi_{1} | \psi\right\rangle\right|=\left|\left\langle \psi_{2}|\psi\right\rangle\right|=\frac{1}{\sqrt{2}}$. Ignoring the global phase, we have $\left|\psi(\lambda)\right\rangle=\frac{1}{\sqrt{2}}\left(|\psi_{1}\rangle+e^{-i\theta(\lambda)}|\psi_{2}\rangle\right)$. Notice that $|\psi(\lambda)\rangle$ is entangled, thus the only possible form is $|\psi(\lambda)\rangle=(|a_{1}a_{2}\cdots a_{n}\rangle+e^{-i\theta(\lambda)}|b_{1}b_{2}\cdots
b_{n}\rangle)/\sqrt{2}$, where $a_{i}\neq b_{i}$, $a_{i},b_{i}\in \{0,1\},\forall i$. Since local optimal achievability is invariant under local unitary transformations, the restriction can be relaxed to $\langle a_{i}|b_{i}\rangle = 0$ for $\forall i$.
\end{proof}}

\section{HOC METHOD}\label{section 4}
Now we introduce another equivalent method called the {``\textit{hierarchy of orthogonality conditions}'' (HOC)} for identifying LM that saturates the QCRB. Although it does not provide an intuitive visual representation of the structure of $M$ like IMP, it is often more efficient in practical calculations.

\hspace{1cm}
\subsection{HOC in the general framework}\label{section 4a}
In the absence of specifying a specific parameter encoding method, HOC can be described as:
{\begin{thm}\label{HOC general theorem}
For an $N$-qubit system labeled by $\mathcal{N} = \{1, 2, \cdots, N\}$, the encoded state $|\psi(\lambda)\rangle$ can saturate the QCRB through LM if and only if there exist $N$ unit vectors $\boldsymbol{n}^{(i)}$ in $\mathbb{R}^{3}$ (allowing multiple vectors to point in the same direction) such that 
\begin{equation}\label{gereral QCRB}
\operatorname{Tr}\left[M \otimes_{j \in \alpha} \boldsymbol{n}^{(j)} \cdot \boldsymbol{\sigma}^{(j)}\right]=0
\end{equation}
holds for all non-empty subsets $\alpha \subseteq \mathcal{N}$, where $M \equiv \left[|\psi(\lambda)\rangle \langle \psi(\lambda)|, L_{\lambda}\right]$, and $\boldsymbol{\sigma}^{(i)} \equiv (X_{i}, Y_{i}, Z_{i})$ are the Pauli operators acting on the $i$-th qubit. When this condition is satisfied, the optimal LM can be chosen as 
\begin{equation}
    E_{x} = \bigotimes_{i=1}^{N} \frac{I^{(i)} + (-1)^{x_{i}} \boldsymbol{n}^{(i)} \cdot \boldsymbol{\sigma}^{(i)}}{2},\quad x_{i} \in \{0,1\}.
\end{equation}
\end{thm}}

\begin{proof}
Firstly, prove the sufficiency. It is straightforward to calculate
\begin{equation}
    \begin{aligned}
E_{x} & =\frac{1}{2^{N}}\left(I^{(1)}+(-1)^{x_{1}} \boldsymbol{n}^{(1)} \cdot \boldsymbol{\sigma}^{(1)}\right) \otimes\left(I^{(2)}+(-1)^{x_{2}} \boldsymbol{n}^{(2)} \cdot \boldsymbol{\sigma}^{(2)}\right)   \\
&\:\:\:\:\:\otimes\cdots\otimes \left(I^{(N)}+(-1)^{x_{N}} \boldsymbol{n}^{(N)}\cdot\boldsymbol{\sigma}^{(N)} \right) \\
& =\frac{1}{2^{N}}\left(I+\sum_{\alpha \subseteq \mathcal{N}, \alpha \neq \emptyset}(-1)^{\sum_{j \in \alpha} x_{j}} \bigotimes_{j \in \alpha} \boldsymbol{n}^{(j)} \cdot \boldsymbol{\sigma}^{(j)}\right).
\end{aligned}
\end{equation}
Considering that $\operatorname{Tr}M=0$, the optimal condition \eqref{ExMEx} can be expressed as
\begin{equation}\label{step 1 alphale 1}
    \sum_{\alpha \subseteq \mathcal{N},|\alpha| \geq 1}(-1)^{\sum_{j \in \alpha} x_{j}} \operatorname{Tr}M\mathcal{A}_{\alpha}=0,\:\forall x_{j},
\end{equation}
where {$\mathcal{A}_{\alpha}\equiv\otimes_{j \in \alpha} \boldsymbol{n}^{(j)} \cdot \boldsymbol{\sigma}^{(j)}$}, and $|\alpha|$ denotes the cardinality of the set $\alpha$. 

Next, we prove the necessity in several steps as follows:

\textbf{Step 1:} Upon setting $x_{N}=0$ and $x_{N}=1$, respectively, while keeping all the remaining $x_{j}$'s fixed, we find
\begin{widetext}
\begin{equation}
    \sum_{\alpha \subseteq \mathcal{N}- \{N\},|\alpha| \geq 1}(-1)^{\sum_{j \in \alpha} x_{j}} \operatorname{Tr}\left(M \mathcal{A}_{\alpha}\right)+\sum_{\beta \subseteq \mathcal{N} -\{N\},|\beta| \ge 0}(-1)^{\sum_{j \in \beta} x_{j}} \operatorname{Tr}\left(M\mathcal{A}_{\beta \cup\{N\}}\right)=0,
\end{equation}
\begin{equation}
    \sum_{\alpha \subseteq \mathcal{N} -\{N\},|\alpha| \geq 1}(-1)^{\sum_{j \in \alpha} x_{j}} \operatorname{Tr}\left(M \mathcal{A}_{\alpha}\right)-\sum_{\beta \subseteq \mathcal{N} -\{N\},|\beta| \ge 0}(-1)^{\sum_{j \in \beta} x_{j}} \operatorname{Tr}\left(M\mathcal{A}_{\beta \cup\{N\}}\right)=0.
\end{equation}
\end{widetext}
Summing over these two equations, we obtain $\sum_{\alpha \subseteq \mathcal{N} -\{N\},|\alpha| \geq 1}(-1)^{\sum_{j \in \alpha} x_{j}} \operatorname{Tr}\left(M \mathcal{A}_{\alpha}\right)=0$.
Iterating the above procedure, we find $\sum_{\alpha = \{1\}}(-1)^{\sum_{j \in \alpha} x_{j}} \operatorname{Tr}\left(M \mathcal{A} _{\alpha}\right)=0$,
i.e., $\operatorname{Tr}\left(M\mathcal{A}_{\alpha}\right)=0$. In a similar manner, one can show that 
\begin{equation}\label{step 1 alpha1}
\operatorname{Tr}\left(M\mathcal{A}_{\alpha}\right),\quad |\alpha|=1.
\end{equation}

\textbf{Step 2:} Substituting Eq. \eqref{step 1 alpha1} into Eq. \eqref{step 1 alphale 1}, we find
\begin{equation}
    \sum_{\alpha \subseteq \mathcal{N},|\alpha| \geq 2}(-1)^{\sum_{j \in \alpha} x_{j}} \operatorname{Tr}\left(M\mathcal{A}_{\alpha}\right)=0.
\end{equation}
Following similar manipulation in $\textbf{Step 1}$, one can first choose certain $x_{j}$ values and set them to be $0$ and $1$ respectively while keeping the remaining ones fixed. Next after summing over the two equations corresponding to $x_{j}=0$ and $x_{j}=1$, the index $j$ is eliminated from the equation, resulting in
\begin{equation}
    \sum_{\alpha \subseteq \mathcal{N}-\{j\} ,|\alpha| \geq 2}(-1)^{\sum_{j \in \alpha} x_{j}} \operatorname{Tr}\left(M \mathcal{A}_{\alpha}\right)=0.
\end{equation}
Iterating this process, we find $\operatorname{Tr}\left(M \mathcal{A}_{\alpha}\right)=0$, $|\alpha|=2$.

Now it is clear that upon reaching $\textbf{Step N}$, we will obtain $\operatorname{Tr}\left(M \mathcal{A}_{\alpha}\right)=0$, where $ |\alpha|=N$,  which concludes the proof.
\end{proof}

\subsection{HOC in the unitary encoding framework}\label{section 4b}
If we constrain the parameter encoding to be unitary encoding, where $U_{\lambda}$
  represents the time-dependent unitary encoding quantum channel that satisfies the Schr\"odinger equation $i\partial_{t}U_{\lambda}(t)=H_{\lambda}(t) U_{\lambda}(t)$, which describes the most general unitary encoding process and includes the common case $U_{\lambda} = e^{-i \lambda t H}$. To make our notation more concise, we define the metrological generator as $G_{\lambda}(t) \equiv i U^{\dagger}_{\lambda}(t) \partial_{\lambda} U_{\lambda}(t)$ \cite{PhysRevLett.98.090401,NatCommun8}. {Note that $\partial_{\lambda}U_{\lambda }\left(t\right)=(-i) \int_{0}^{t} \partial_{\lambda}\left (H_{\lambda}\left(s\right) U_{\lambda}(s)\right)d s$, which implies $G_{\lambda}(t=0)=0$. Therefore, $G_{\lambda}(t)$ can be expanded into 
\begin{equation}
    \begin{aligned}
G_{\lambda}(t)&=G_{\lambda}(t=0)+\int_{0}^{t}ds \partial_{s}(iU^{\dagger}_{\lambda}\partial_{\lambda}U_{\lambda})\\
&=\int_{0}^{t}ds \partial_{s}(iU^{\dagger}_{\lambda}\partial_{\lambda}U_{\lambda})\\
&=\int_{0}^{t}ds  U_{\lambda}^{\dagger}(s) \partial_{\lambda} H_{\lambda}(s) U_{\lambda}(s).
\end{aligned}
\end{equation}}

Now we are ready to provide the HOC for the LM saturation of the QCRB in the unitary phase encoding case.
\begin{thm}
The time-dependent unitary phase-encoding parameter estimation model $|\psi(\lambda)\rangle = U_{\lambda}(t)|\psi_{0}\rangle$ can saturate the Quantum Cram\'er-Rao Bound through local measurements  if and only if 
\begin{equation}\label{COVNGlambdat}
\operatorname{Cov}\left(\mathcal{A}_{\alpha}^{(H)}(t) G_{\lambda}(t)\right)_{\left|\psi_{0}\right\rangle} = 0
\end{equation}
holds for $\forall \alpha \subseteq \mathcal{N} \equiv \{1, 2, \cdots, N\}$, where  $\mathcal{A}_{\alpha}^{(H)}(t)\equiv U^{\dagger}_{\lambda}(t)\mathcal{A}_{\alpha}U_{\lambda}(t)$ is  the Heisenberg evolution of $\mathcal{A}_{\alpha} \equiv \otimes_{n \in \alpha} \boldsymbol{n}^{(j)} \cdot \boldsymbol{\sigma}^{(j)}$ and $\operatorname{Cov}(A B)_{\left|\psi_{0}\right\rangle} \equiv \frac{1}{2}\langle AB+BA\rangle_{\left|\psi_{0}\right\rangle}-\langle A\rangle_{|\psi_{0}\rangle}\langle B\rangle_{\left|\psi_{0}\right\rangle}$.
\end{thm}

\begin{proof}
Considering that $\rho(\lambda)=U_{\lambda} \rho_{0} U_{\lambda}^{\dagger}$, where $\rho_{0}=|\psi_{0}\rangle \langle \psi_{0}|$, it is straightforward to compute
\begin{equation}
    \partial_{\lambda} \rho(\lambda)=\partial_{\lambda} U_{\lambda} \rho_{0} U_{\lambda}^{\dagger}+U_{\lambda} \rho_{0} \partial_{\lambda} U_{\lambda}^{\dagger}=-\mathrm{i} U_{\lambda}[G_{\lambda}, \rho_{0}] U_{\lambda}^{\dagger},
\end{equation}
where we used $\partial_{\lambda}U_{\lambda}U_{\lambda}^{\dagger}+U_{\lambda}\partial_{\lambda}U_{\lambda}^{\dagger}=0$. Thus
\begin{equation}\label{unitary M}
    M=2[\rho(\lambda), L_{\lambda}]=2[\rho(\lambda), \partial_{\lambda} \rho(\lambda)]=-2i U_{\lambda}[\rho_{0},[G_{\lambda}, \rho_{0}]] U_{\lambda}^{\dagger}.
\end{equation}
Substituting Eq. \eqref{unitary M} into Eq. \eqref{gereral QCRB}, we arrive at $\operatorname{Tr}\left([\rho_{0},[G_{\lambda}, \rho_{0}]] \mathcal{A}_{\alpha}^{(H)}\right)=0$, which can be further rewritten as
\begin{equation}
    \begin{aligned}
0&=\operatorname{Tr}\left(\left[G_{\lambda}, \rho_{0}\right]\left[\mathcal{A}_{\alpha}^{(H)}, \rho_{0}\right]\right)\\
&=2\operatorname{Tr}\left(G_{\lambda}\rho_{0}\mathcal{A}_{\alpha}\rho_{0}\right)-\operatorname{Tr} \left(\mathcal{A}_{\alpha}^{(H)}G_{\lambda}\rho_{0}\right)-\operatorname{Tr} \left(G_{\lambda}\mathcal{A}_{\alpha}^{(H)}\rho_{0}\right)\\
&=2\langle G_{\lambda}\rangle_{|\psi_{0}\rangle}\langle \mathcal{A}_{\alpha}^{(H)}\rangle_{|\psi_{0}\rangle}-\langle\mathcal{A}_{\alpha}^{(H)} G_{\lambda}+G_{\lambda}\mathcal{A}_{\alpha}^{(H)} \rangle_{\left|\psi_{0}\right\rangle}.
\end{aligned}
\end{equation}
\end{proof}

\subsection{The application of HOC}\label{section 4c}
In this subsection, we will provide some illustrative examples to demonstrate the effectiveness of the HOC method in studying the saturation of the QCRB for LM. This includes a {metrological example of locally optimal achievable cases for a three-qubit W state and a class of $N$-qubit W states}, as well as a locally optimal unreachable example for a three-qubit W state.

{\textit{Example 1 (locally optimal achievable).} The first locally optimal achievable example is the phase estimation model for a three-qubit system, defined by $|\psi(\lambda)\rangle = e^{-i\lambda H}|\psi_{0}\rangle$ (we set $t=1$ for convenience), where the probe state is $|\psi_{0}\rangle =|W^{(3)}\rangle\equiv  \frac{1}{\sqrt{3}}(|100\rangle + |010\rangle + |001\rangle)$ and the encoding Hamiltonian is $H = X_{1}X_{2} + X_{2}X_{3}$. We observe that for $\forall \lambda$, there exists an optimal LM $E_{x}=\bigotimes_{j=1}^{3}\frac{I^{(i)}+(-1)^{x_{i}}(\cos \alpha_{j}X_{j}+\sin \alpha_{j} Y_{j})}{2}$ that saturates the QCRB.
To see this and clarify the value of $\alpha_{j}$, it is observed that for  $U_{\lambda} = e^{-i\lambda H}$, ignoring an irrelevant global factor, which does not affect the fulfillment of the optimal condition \eqref{gereral QCRB}, the $M$ matrix can be expressed as 
\begin{equation}
    M = \left\{H, \rho(\lambda)\right\} - 2\langle\psi_{0}|H| \psi_{0}\rangle\rho(\lambda)= \left\{h, \rho(\lambda)\right\}
\end{equation}
where $\rho(\lambda)=|\psi(\lambda)\rangle \langle \psi(\lambda)|$, $h\equiv H-\langle \psi_{0}|H|\psi_{0}\rangle$, and $\{A,B\}\equiv AB+BA$. Direct calculation shows that for any $j\in \{1,2,3\}$, $\boldsymbol{m}^{(j)}$ defined by $\boldsymbol{m}^{(j)}\cdot \boldsymbol{\sigma}^{(j)}=\operatorname{Tr}_{\bar{j}}\left(M\right)$ satisfies $\boldsymbol{m}^{(j)}\propto  \hat{\boldsymbol{z}}$, where $\boldsymbol{\sigma}^{(j)}\equiv(X_{j},Y_{j},Z_{j})$. Therefore, when $|\alpha|=1$, the optimal condition \eqref{gereral QCRB} can be expressed as $\operatorname{Tr}\left[M\left(\boldsymbol{n}^{(j)}\cdot \boldsymbol{\sigma}^{(j)}\right)\right]=0$, indicating that $\boldsymbol{n}^{(j)}$ must be on X-Y plane, i.e., $\boldsymbol{n}^{(j)}= \cos \alpha_{j}\hat{\boldsymbol{x}}+ \sin \alpha_{j}\hat{\boldsymbol{y}}$.
For the case of $|\alpha|=2$, we have
\begin{equation}
\operatorname{Tr} M \mathcal{A}_{j k}=\operatorname{Tr}\left\{h, \mathcal{A}_{j k}\right\}\rho(\lambda)=2\left(\cos \alpha_{j}, \sin \alpha_{j}\right) T_{j k}\binom{\cos \alpha_{k}}{\sin \alpha_{k}},
\end{equation}
where 
\begin{equation}
    T_{j k}\equiv\frac{1}{2}\left(\begin{array}{cc}
\left\langle\left\{h, X_{j k}\right\}\right\rangle_{\lambda} & \left\langle\left\{h, X_{j} Y_{k}\right\}\right\rangle_{\lambda} \\
\left\langle\left\{h, Y_{j} X_{k}\right\}\right\rangle_{\lambda} & \left\langle\left\{h, Y_{j k}\right\}\right\rangle_{\lambda}
\end{array}\right),
\end{equation}
where $\langle A \rangle_{\lambda} \equiv \operatorname{Tr}A\rho(\lambda)$. As $-Z_{123}\equiv -Z_{1}Z_{2}Z_{3}$ is a symmetry, i.e., $[H,Z_{123}]=0$ and stabilizes the state $|\psi_{\lambda}\rangle $ but anticommutes with $N_{123}$, we have $\operatorname{Tr} M \mathcal{A}_{123}=\left\langle\left\{h, \mathcal{A}_{123}\right\} \right\rangle_{\lambda}
=-\left\langle\left\{h, \mathcal{A}_{123}\right\} Z_{123}\right\rangle_{\lambda}=\left\langle Z_{123}\left\{h, \mathcal{A}_{123}\right\}\right\rangle_{\lambda}=0$. Thus we have only to check
whether there are $\left\{|\alpha_{j}\rangle\equiv \left(\cos \alpha_{j}, \sin \alpha_{j}\right)\right\}_{j=1}^{3}$ such that $\langle\alpha_{j}|T_{j k}| \alpha_{k}\rangle=0$ for all $j>k$. Note that $T_{jk}=T^{\dagger}_{kj}$. By swapping symmetry $V_{13}$ we have $T_{12}=T_{32}$, $T_{13}=T_{31}$. Direct calculation shows
\begin{equation}
    T_{12}=\frac{1}{9}\left(\begin{array}{cc}
7 & 12 \sin 2 \lambda\cos 2 \lambda \\
-2 \sin 2 \lambda & -5 \cos 2 \lambda
\end{array}\right),
\end{equation}
\begin{equation}
    T_{13}=\frac{1}{9}\left(\begin{array}{cc}
4 & 5 \sin 2 \lambda \\
5 \sin 2 \lambda & 4 \sin ^{2} 2 \lambda-8 \cos ^{2} 2 \lambda
\end{array}\right),
\end{equation}
\begin{equation}
\operatorname{Det} T_{13}=-\frac{1}{81}\left(32 \cos ^{2} 2 \lambda+9 \sin ^{2} 2 \lambda\right)<0.
\end{equation}
By denoting 
\begin{equation}
    T\equiv T_{21} Y T_{13} Y  T_{32}= \begin{pmatrix}
T_{xx} & T_{xy} \\
T_{yx} & T_{yy}
\end{pmatrix},
\end{equation}
where $Y$ is the Pauli-Y matrix, we have
\begin{equation}
    \left\{
\begin{array}{l}
T_{xx} = -\frac{4}{729} (5 + 93 \cos 4\lambda), \\
T_{xy} =T_{yx}= -\frac{\left(-107 + 564 \cos4\lambda\right)\sin4\lambda}{1458},  \\
T_{yy} =\frac{4}{729} \left(118 - 147 \cos4\lambda + 54 \cos8\lambda\right) \cos^2 2\lambda, \\
\end{array}
\right.
\end{equation}
\begin{equation}\label{DetT le 0}
    \operatorname{Det} T=T_{xx}T_{yy}-T_{xy}T_{yx}=\left(\operatorname{Det} T_{12}\right)^{2} \operatorname{Det} T_{13} \leq 0,
\end{equation}
which infers that there exists a unit vector $|\alpha_{2}\rangle $ such that $\left\langle\alpha_{2}|T| \alpha_{2}\right\rangle=0$.
To see this, let $|\alpha_{2}\rangle =(\cos \alpha_{2},\sin \alpha_{2})$ then the above equation becomes $T_{x x} \cos ^{2} \alpha_{2}+T_{y y} \sin ^{2} \alpha_{2}+2 T_{x y} \sin \alpha_{2} \cos \alpha_{2}=0$,
which is equivalent to $T_{x x} s^{2}+T_{y y}+2 T_{x y} s=0$ with $s\equiv \cot \alpha_{2}$. Obviously, it has a real solution
\begin{equation}
    s=\frac{-T_{x y} \pm \sqrt{T_{x y}^{2}-T_{x x} T_{y y}}}{T_{x x}}
\end{equation}
if and only if $T_{x y}^{2}-T_{x x} T_{y y} \geq 0$, combining Eq. \eqref{DetT le 0}, this condition is exactly satisfied.}

{Let $\left|\alpha_{1}\right\rangle \propto i Y \cdot T_{12}\left|\alpha_{2}\right\rangle$ and $\left|\alpha_{3}\right\rangle \propto i Y \cdot T_{32}\left|\alpha_{2}\right\rangle$ (with $T_{12}|\alpha_{2}\rangle \neq 0$), we have
\begin{equation}
    \begin{aligned}
        \left\langle\alpha_{2}\left|T_{21}\right| \alpha_{1}\right\rangle &=\left\langle\alpha_{2}\left|T_{21} \cdot Y \cdot T_{12}\right| \alpha_{2}\right\rangle=0, \\
       \left\langle\alpha_{2}\left|T_{23}\right| \alpha_{3}\right\rangle &=\left\langle\alpha_{2}\left|T_{23} \cdot Y \cdot T_{32}\right| \alpha_{2}\right\rangle=0,\\
\left\langle\alpha_{1}\left|T_{13}\right| \alpha_{3}\right\rangle&=\left\langle\alpha_{2}\left|T_{21} \cdot Y\cdot T_{13} \cdot Y \cdot T_{32}\right| \alpha_{2}\right\rangle=0.
   \end{aligned}
\end{equation}
This proves all the cases where $T_{k2}|\alpha_{2}\rangle \neq 0$, $k=1,3$. Otherwise, $T_{12}\left|\alpha_{2}\right\rangle=T_{32}\left|\alpha_{2}\right\rangle=0$, let $\left\{\left|\alpha_{1}\right\rangle,\left|\alpha_{3}\right\rangle\right\}$ be two eigenvectors of $T_{13}=T_{13}^{\dagger}\neq I$ corresponding to different eigenvalues, we then obtain
$\left\langle\alpha_{1}\left|T_{13}\right| \alpha_{3}\right\rangle=\lambda_{\alpha_{3}}\left\langle\alpha_{1}| \alpha_{3}\right\rangle=0$.
It is evident that $\left\langle\alpha_{2}\left|T_{21}\right| \alpha_{1}\right\rangle=0$ and $\left\langle\alpha_{3}\left|T_{32}\right| \alpha_{2}\right\rangle=0$. By identifying $|\alpha_{j}\rangle = \left(\cos \alpha_{j}, \sin \alpha_{j}\right)$, we have the desired local optimal measurement, and the specific form can be found in Appendix \ref{app 3}.}

{\textit{Example 2 (locally optimal reachable).} The second class of locally optimal reachable $N$-qubit examples is the phase estimation model $|\psi(\lambda)\rangle = e^{-i\lambda H}|\psi_{0}\rangle $, where the probe state 
\begin{equation}
    |\psi_{0}\rangle = |\widetilde{W}^{(N)}\rangle \equiv \frac{1}{\sqrt{N}} \sum_{i=1}^{N}(-1)^{S(i)}|i\rangle=\left(\prod_{i:S(i)=1}Z_{i}\right)|W^{(N)}\rangle
\end{equation}
represents a class of generalized W states, which can be obtained from the W states $|W^{(N)}\rangle$ by applying local $Z$ gates. Here, $S(i)$ denotes the $i$-th element of a certain $N$-dimensional binary (0 or 1) vector $S$, and $|i\rangle \equiv |0\cdots 0 \underset{i\text{-qubit}}{1} 0 \cdots 0\rangle$. The encoding Hamiltonian is $H = \sum_{j=1}^{N-1}\left(X_{j} X_{j+1}+Y_{j} Y_{j+1}\right)$. We have discovered that for $\forall N =2k+1$, where $k\in \mathbb{Z},k\ge 1$, selecting a specific vector $S = \tilde{S}$
  enables the existence of a LM that saturates the QCRB for the above model at $\lambda=0$.  The above examples provide another class of locally optimal reachable multi-qubit metrological models that differs from the GHZ state. The proof of saturation, as well as the specific forms of $\tilde{S}$
  and the optimal LM, can be found in Appendix \ref{N qubit Generalized W State}.}

\textit{{Example 3 (locally optimal unachievable).}} The third locally optimal unreachable example is also in the context of the three-qubit phase estimation model $|\psi_{\lambda}\rangle = e^{-i\lambda H}|\psi_{0}\rangle$, where the probe state is $|\psi_{0}\rangle =|W^{(3)}\rangle =\frac{1}{\sqrt{3}}(|100\rangle + |010\rangle + |001\rangle)$ and the encoding Hamiltonian is $H = X_{1}X_{2} + X_{2}X_{3} + Y_{1}Y_{2} + Y_{2}Y_{3}$. It is observed that at $\lambda = 0$, there does not exist a LM that saturates the QCRB, due to the fact that the HOC cannot be simultaneously satisfied. Specifically, considering that the vector $\boldsymbol{m}^{(j)}$ defined by $\boldsymbol{m}^{(j)}\cdot \boldsymbol{\sigma}^{(j)}=\operatorname{Tr}_{\bar{j}}\left(M\right)$ satisfies $\boldsymbol{m}^{(j)}\propto  \hat{\boldsymbol{z}}$, Eq. \eqref{COVNGlambdat} implies that $\mathcal{A}_{j}=X_{j}\cos \beta_{j}+Y_{j}\sin \beta_{j}$ can only be constrained to the X-Y plane. Substituting $\mathcal{A}_{j}$
  and $\mathcal{A}_{k}$, $j\neq k$, into Eq. \eqref{COVNGlambdat} for $|\alpha|=2$, simplifying the expression gives $\cos \left(\beta_{j}-\beta_{k} \right)=0$, which are clearly incompatible with each other.

{\section{Conclusion}\label{conclusion}
In this work, we focus on achieving the QCRB in parameter estimation through local measurements. For qubit systems, we demonstrate that if there exists a local measurement that saturates the QCRB, it is always possible to construct another rank-$1$ local projective measurement that also saturates the QCRB, indicating that considering rank-$1$ local measurements is sufficient. We propose two necessary and sufficient methods, ``\textit{iterative matrix partition}'' (IMP) and ``\textit{hierarchy of orthogonality conditions}'' (HOC), to address the saturation of the QCRB for pure qubit states. IMP can quickly identify locally optimal measurements based on the structure of the $M$ matrix; typical examples include the GHZ state when $M$ is diagonal and scenarios when it is zero-diagonal. Furthermore, IMP can universally reconstruct the existing conclusion of LMCC saturating the QCRB. For HOC, we present formulations under general encoding and unitary phase encoding, providing examples of two types of locally optimal achievable cases for $N=3$ and $N \ge 3$ around W states, as well as one example of a locally optimal unachievable case.}

{We anticipate the near-term implementation of these methods in noisy intermediate scale quantum devices \cite{PhysRevLett.123.040501,npj7170,nature7170}. Future research directions include extending the methods to qudits \cite{Len2022}, continuous variable systems \cite{PhysRevA.102.052601}, and qubit-cavity systems \cite{PhysRevB.109.L041301}, as well as exploring their application in entanglement detection \cite{PhysRevA.85.022321,PhysRevA.85.022322,PhysRevA.88.014301} and spin-squeezing \cite{MA201189,Toth_2014,PhysRevA.47.5138}. Additionally, further investigations will be conducted on the impact of decoherence and other related factors.}

\section*{Acknowledgments}
We thank Sisi Zhou for useful comments on the manuscript. JY was funded by the Wallenberg Initiative on Networks and Quantum Information (WINQ). HLS was supported by the European Commission through the H2020 QuantERA ERA-NET Cofund in Quantum Technologies project ``MENTA'' and the NSFC key grants No. 12134015 and No. 92365202. SXY was supported by Key-Area Research and Development Program of Guangdong Province Grant No. 2020B0303010001.

\appendix
{\section{The Proof of Observation \ref{observation}}\label{Universality of Two-Qubit Pure States}
To prove sufficiency, Note that the fact that both diagonal elements of a $\mathit{2\times 2}$ Hermitian (or anti-Hermitian) matrix $\mathcal{A}_{j}$ are zero implies that $\boldsymbol{a}_{j}$ lies in the X-Y plane or that $\operatorname{Tr}\mathcal{A}_{j}Z=0$. If all $\boldsymbol{a}_{j}$ are on the same plane that is orthogonal to a unit normal vector $\hat{\boldsymbol{n}}$ then a rotation from $\hat{\boldsymbol{n}}$ to $\hat{\boldsymbol{z}}$ suffices to make the diagonal entries of all $\mathcal{A}_{j}$ vanishing.  The corresponding unitary reads
\begin{equation}
    U=\exp \left\{-i \frac{\theta}{2} \frac{\hat{\boldsymbol{n}} \times \hat{\boldsymbol{z}}}{|\hat{\boldsymbol{n}} \times \hat{\boldsymbol{z}}|} \cdot \boldsymbol{\sigma}\right\}
\end{equation}
with $\cos \theta=\hat{\boldsymbol{n}}\cdot \hat{\boldsymbol{z}}$ as
\begin{equation}\label{UnsigmaUdaggerZ}
    \begin{aligned}
U \hat{\boldsymbol{n}} \cdot \boldsymbol{\sigma}U^{\dagger}&=\hat{\boldsymbol{n}} \cdot \boldsymbol{\sigma}\left(\cos \theta+i \frac{\hat{\boldsymbol{n}} \times \hat{\boldsymbol{z}}}{|\hat{\boldsymbol{n}} \times \hat{\boldsymbol{z}}|} \cdot \boldsymbol{\sigma} \sin \theta\right)\\
&=(\hat{\boldsymbol{n}} \cos \theta+(\hat{\boldsymbol{n}} \times \hat{\boldsymbol{z}})\times \hat{\boldsymbol{n}} ) \cdot \boldsymbol{\sigma}\\
&=Z,
\end{aligned}
\end{equation}
where we used the relation $\exp (i \gamma \boldsymbol{u} \cdot \boldsymbol{\sigma})=\cos \gamma I+i \sin \gamma \boldsymbol{u} \cdot \boldsymbol{\sigma}$ and $\left(\boldsymbol{u}\cdot \boldsymbol{\sigma}\right)\left(\boldsymbol{v}\cdot \boldsymbol{\sigma}\right)=\boldsymbol{u}\cdot \boldsymbol{v}+i\left(\boldsymbol{u}\times \boldsymbol{v}\right)\cdot\boldsymbol{\sigma}$. By utilizing Eq. \eqref{UnsigmaUdaggerZ}, one can readily show
\begin{equation}
    \operatorname{Tr}\left(U \mathcal{A}_{j} U^{\dagger} Z\right)=\operatorname{Tr}\left(\mathcal{A}_{j} U^{\dagger} Z U\right)=\operatorname{Tr}\left[\left(\boldsymbol{a}_{j} \cdot \boldsymbol{\sigma}\right)(\hat{\boldsymbol{n}} \cdot \boldsymbol{\sigma})\right]=0.
\end{equation}
The proof of necessity can be derived by reversing the above argument.}

{\section{Local Optimal Measurement of Example 1}\label{app 3}
The local optimal measurement for Example 1 is given by
\begin{equation}
E_{\pm \pm \pm }=\bigotimes_{i=1}^{3}\frac{1}{2}\left[I^{(i)}\pm (\cos \alpha_{i}X_{i}+\sin\alpha_{i}Y_{i})\right],
\end{equation}
where
\begin{widetext}
\begin{equation}
 \cot \alpha_{2}= -\frac{\sqrt{2} \sqrt{(41 + 23 \cos4\lambda)(29 \cos2\lambda + 6 \cos6\lambda)^2} - 107 \sin4\lambda + 282 \sin8\lambda}{8 (5 + 93 \cos4\lambda)},
\end{equation}
\begin{equation}
   \cos\alpha_{1}= \cos \alpha_{3}=\frac{  -5 \sin\alpha_{2} \cos2\lambda - 2\cos\alpha_{2} \sin2\lambda }{\sqrt{(5\sin \alpha_{2} \cos2\lambda + 2\cos \alpha_{2} \sin2\lambda)^2 + (7\cos \alpha_{2} + 6\sin \alpha_{2} \sin4\lambda)^2}},
\end{equation}
\begin{equation}
   \sin \alpha_{1}= \sin \alpha_{3}=\frac{  -7 \cos\alpha_{2}  - 6\sin\alpha_{2} \sin4\lambda }{\sqrt{(5\sin \alpha_{2} \cos2\lambda + 2\cos \alpha_{2} \sin2\lambda)^2 + (7\cos \alpha_{2} + 6\sin \alpha_{2} \sin4\lambda)^2}}.
\end{equation}
\end{widetext}
Precise calculations show that $T_{12} |\alpha_{2}\rangle \neq 0$ holds universally, therefore it follows that $\alpha_{1} = \alpha_{3}$ always.}

{\section{The Proof of Example 2}\label{N qubit Generalized W State}}
{For $\forall i \in \{1, 2, \ldots, N-2\}$, if we impose the constraint $S(i) + S(i+2) = 1$ on $S$, then the expectation value of $H$ with respect to the probe state $|\psi_{0}\rangle$ can be simplified to
\begin{equation}
    \langle H\rangle_{\left|\psi_{0}\right\rangle}=\frac{2}{N}\left[(-1)^{S(1)+S(2)}+(-1)^{S(N-1)+S(N)}\right].
\end{equation}
If we further constrain $S$ to ensure $\langle H\rangle_{\left|\psi_{0}\right\rangle}=0$, then $S$ should satisfy $(-1)^{S(1)+S(2)}+(-1)^{S(N-1)+S(N)}=0$. Combining the above two constraint conditions, for all $N\ge 3$ and $N$ being odd, we can set $\tilde{S}(1)=1$, $\tilde{S}(2)=0$, and determine $\tilde{S}(3)$ up to $\tilde{S}(N)$ by repeatedly using the condition $\tilde{S}(i)+\tilde{S}(i+2)=1$. This choice of $\tilde{S}$ satisfies the equations $(-1)^{\tilde{S}(1)+\tilde{S}(2)} + (-1)^{\tilde{S}(N-1)+\tilde{S}(N)} = 0$ and $(-1)^{\tilde{S}(2)+\tilde{S}(N)} + (-1)^{\tilde{S}(1)+\tilde{S}(N-1)} = 0$, which can help prove the saturation of the QCRB in subsequent analysis. Taking into account $\langle H\rangle_{\left|\psi_{0}\right\rangle}=0$ and $\lambda=0$, the condition \eqref{COVNGlambdat} can be written as
\begin{equation}\label{lambda0COV}
     \left\langle \mathcal{A}_{\alpha} H+H\mathcal{A}_{\alpha}\right\rangle_{\left|\psi_{0}\right\rangle}=0, \quad \forall \alpha \subseteq \mathcal{N}.
\end{equation}
Below we prove that when the LM is chosen as $\boldsymbol{n}^{(1)}=\boldsymbol{n}^{(N)}=\boldsymbol{\hat{x}}$ and  $\boldsymbol{n}^{(2)}=\cdots =\boldsymbol{n}^{(N-1)}=\boldsymbol{\hat{z}}$, Eq. \eqref{lambda0COV} holds.} 

{When $|\{1,N\} \cap \alpha| = 0$, we have
\begin{equation}
    \begin{aligned}
\langle H\mathcal{A}_{\alpha}\rangle _{|\psi_{0}\rangle}&=\frac{2}{\sqrt{N}}\left[ (-1)^{\tilde{S}(2)}\langle1|Z_{\alpha}|\psi_{0}\rangle+ (-1)^{\tilde{S}(N-1)}\langle N|Z_{\alpha}|\psi_{0}\rangle\right]\\
&=\frac{2}{\sqrt{N}} (-1)^{\tilde{S}(2)}\langle1|\psi_{0}\rangle+\frac{2}{\sqrt{N}} (-1)^{\tilde{S}(N-1)}\langle N|\psi_{0}\rangle\\
&\propto\left[(-1)^{\tilde{S}(1)+\tilde{S}(2)}+(-1)^{\tilde{S}(N-1)+\tilde{S}(N)}\right]=0,
\end{aligned}
\end{equation}
therefore, $\left\langle \mathcal{A}_{\alpha} H+H\mathcal{A}_{\alpha}\right\rangle_{\left|\psi_{0}\right\rangle}=0$.}

{When $|\{1,N\} \cap \alpha| = 1$, since at this point, $\operatorname{span}\{|1\rangle, |2\rangle, \ldots, |N\rangle\}$ lies in the null space of $\mathcal{A}_{\alpha}|\psi_{0}\rangle$, thus $\left\langle\psi_{0}\left|\mathcal{A}_{\alpha} H\right| \psi_{0}\right\rangle=\left\langle\psi_{0}\left|H \mathcal{A}_{\alpha}\right| \psi_{0}\right\rangle=0$.}

{When $|\{1,N\} \cap \alpha| = 2$, we obtain
\begin{equation}
\begin{aligned}
\langle H\mathcal{A}_{\alpha}\rangle _{|\psi_{0}\rangle}&=\frac{2}{\sqrt{N}}\left[ (-1)^{\tilde{S}(2)}\langle1|\mathcal{A}_{\alpha}|\psi_{0}\rangle+ (-1)^{\tilde{S}(N-1)}\langle N|\mathcal{A}_{\alpha}|\psi_{0}\rangle\right]\\
&=\frac{2}{\sqrt{N}} (-1)^{\tilde{S}(2)}\langle1|\psi_{0}\rangle+\frac{2}{\sqrt{N}} (-1)^{\tilde{S}(N-1)}\langle N|\psi_{0}\rangle\\
&\propto \left[(-1)^{\tilde{S}(2)+\tilde{S}(N)}+(-1)^{\tilde{S}(1)+\tilde{S}(N-1)}\right]=0.
\end{aligned}
\end{equation}
Based on the above, for $\forall N \ge 3$, as long as $S = \tilde{S}$, there always exists an LM given by 
\begin{equation}
E_{x} = \bigotimes_{i=1}^{N} \frac{I^{(i)} +(-1)^{x_{i}} \boldsymbol{n}^{(i)} \cdot \boldsymbol{\sigma}^{(i)}}{2}, \quad x_{i}\in \{0,1\},
\end{equation}
where $\boldsymbol{n}^{(1)} = \boldsymbol{n}^{(N)} = \boldsymbol{\hat{x}}$ and $\boldsymbol{n}^{(2)} = \cdots = \boldsymbol{n}^{(N-1)} = \boldsymbol{\hat{z}}$, saturating the QCRB.}

\bibliography{dsg} 
\end{document}